%% file: body.tex
\pdfoutput=1

\documentclass[11pt]{article}
\usepackage[affil-it]{authblk}
\usepackage{tikz}
\usetikzlibrary{
  cd,
  math,
  decorations.markings,
  decorations.pathreplacing,
  positioning,
  arrows.meta,
  circuits.logic.US,
  shapes,
  calc,
  fit,
  trees,
  quotes}
\usetikzlibrary{decorations.pathmorphing}
\usetikzlibrary{decorations.markings}
\usetikzlibrary{decorations.pathreplacing}
\usetikzlibrary{arrows}
\usetikzlibrary{shapes.geometric}
\usepackage{amsmath}
\usepackage[utf8]{inputenc}
\usepackage[english]{babel}
\usepackage{freetikz}

\usepackage{amssymb}
\usepackage{enumitem}
\usepackage{bbm}
\usepackage{relsize}

\pgfdeclarelayer{edgelayer}
\pgfdeclarelayer{nodelayer}
\pgfsetlayers{edgelayer,nodelayer,main}
\tikzstyle{none}=[inner sep=0pt]
\usepackage{tkz-graph}
\usetikzlibrary{shapes.geometric}%

\usepackage{tikz}
\usepackage{mathtools}
\usepackage{amsthm}
\usepackage{amsmath}
\usepackage{amssymb}
\usepackage{proof}

\usepackage{listings}
\usepackage{braket}
\usepackage{color}
\usepackage{adjustbox}
\usepackage{physics}

\usetikzlibrary{matrix,arrows,decorations.pathmorphing}

\newcommand{\s}{\enspace}
\newcommand{\sub}{\subseteq}

\renewcommand{\tt}[1]{\mathtt{#1}}
\renewcommand{\bf}[1]{\mathbf{#1}}

\renewcommand{\cal}[1]{\mathcal{#1}}

\newcommand{\bb}[1]{\mathbb{#1}}

\newcommand{\N}{\bb{N}}
\newcommand{\ssxv}{\mathbb{S} \times (\mathbb{S}^X)^{V}}

\newtheorem{lemma}{Lemma}[section]
\newtheorem{proposition}{Proposition}[section]
\newtheorem{corollary}{Corollary}[lemma]
\newtheorem{definition}{Definition}[section]
\newtheorem{example}{Example}[definition]

\usepackage{xspace}
\usepackage{color}
\def\bR{\begin{color}{red}}
\def\e{\end{color}\xspace}

\title{Incremental Monoidal Grammars}

\author{Dan Shiebler\textsuperscript{1}, Alexis Toumi\textsuperscript{1}, Mehrnoosh Sadrzadeh\textsuperscript{2} \\ 
University of Oxford\textsuperscript{1}, UCL\textsuperscript{2}}

\date{\today}

\begin{document}
\maketitle

\input{0-abstract}
\input{1-introduction}
\input{2-preliminaries}
\input{3-monoidal-grammar}

\input{4-incremental-monoidal}
\input{5-r-monoidal-grammars}
\input{999-future-works}

\footnotesize
\bibliography{body.bib}
\bibliographystyle{alpha}
\end{document}

%% file: 0-abstract.tex

\begin{abstract}
  In this work we define formal grammars in terms of free monoidal categories, along with a functor from the category of formal grammars to the category of automata.
  Generalising from the Booleans to arbitrary semirings, we extend our construction to weighted formal grammars and weighted automata.
  This allows us to link the categorical viewpoint on natural language to the standard machine learning notion of probabilistic language model.
\end{abstract}

%% file: 1-introduction.tex

Recurrent neural networks and probabilistic language models \cite{BengioEtAl03} have become standard tools in the natural language processing (NLP) community.
These networks are inherently incremental, scanning through the sequence of words they are given and updating their prediction for what comes next.
Despite their practical success on hard language tasks such as translation and question answering, the structure underlying these machine learning models is yet poorly understood, and they are generally only used as black boxes.

On the other hand, the categorical compositional distributional (DisCoCat) models of Coecke et al. \cite{ClarkEtAl08,ClarkEtAl10,CoeckeEtAl13} use grammatical structure, explicitly encoded as string diagrams in a free monoidal category, to compute natural language semantics.
DisCoCat models have received experimental support on small-scale tasks  \cite{GrefenstetteSadrzadeh11,KartsaklisEtAl12,KartsaklisEtAl13}, but the extra mathematical structure makes them hard to scale to the billions of words used in the training of modern NLP models.

In this work, we aim to bridge this gap by constructing a functor which sends formal grammars (encoded as monoidal categories) to a state automaton that parses the grammar in an incremental way, reading one word at a time and updating the set of possible parsings.
In section \ref{preliminaries} we introduce some core definitions and preliminaries. In section \ref{grammars} we introduce monoidal grammars and $\mathbb{S}$-monoidal grammars for an arbitrary semiring $\mathbb{S}$. In section \ref{incrementalgrammars} we define the functor $\cal{I}$, from the category of $\mathbb{S}$-monoidal grammars to coalgebras of the weighted automata functor $\mathcal{W}(X) = \ssxv$. In section \ref{rgrammars}, we describe an algorithm for learning $\mathbb{R}_{\geq 0}$-monoidal grammars from a probabilistic language model.
This paper represents a work in progress, and we conclude with a discussion of future work.

%% file: 2-preliminaries.tex

\section{Preliminaries}\label{preliminaries}

\subsection{Formal Grammars}

Given a pair of sets $X$ and $Y$, we write $X \times Y$ and $X + Y$ for their Cartesian product and disjoint union respectively.
We use the Kleene star $X^* = \coprod_{n \in \N} X^n$ to denote the free monoid, with concatenation as product and the empty sequence $\epsilon \in X^0$ as unit.
Similarly, we write $X^+ = \coprod_{n > 0} X^n$ for the free semigroup (i.e. non-empty sequences).

If we fix a finite set of words $V$ called the \emph{vocabulary}, then a the language of a formal grammar over this set of words is a subset $L \sub V^\star$: the set of sequences of words that form grammatical sentences accoridng to the rules of the grammar.
A well-known class of formal grammars is \emph{context-free grammars}, introduced by Chomsky as the first level of his complexity-theoretic hierarchy \cite{Chomsky56, Chomsky57, Chomsky65}.

\begin{definition}
  A \emph{context-free grammar} is given by a 4-tuple $(V, X, R, s)$ where:
  \begin{itemize}
    \item $V$ is a finite set called the \emph{vocabulary} (or terminal symbols),
    \item $X$ is a finite set of \emph{syntactic variables} (or nonterminal symbols),
    \item $R \sub X \times (V + X)^*$ is a finite set of \emph{production rules},
    \item $s \in X$ is a syntactic variable called the \emph{start symbol}.
  \end{itemize}
\end{definition}

We define the language $L(G) \sub V^\star$ of a context-free grammar $G = (V, X, R, s)$ as follows: We first construct a binary relation $(\to_R) \sub (V + X)^* \times (V + X)^*$ where for any pair of strings $s, t \in (V + X)^*$ we have $s \to_R t$ if and only if there is some production rule $(x, v) \in R$ and a pair $u, w \in (V + X)^*$ such that $s = u x w$ and $t = u v w$. We then compute its reflexive transitive closure $\to_R^*$ and define $L(G) = \set{ u \in V^* \ \vert \ s \to_R^* u}$.
Note that this corresponds precisely to the preordered monoid generated by the production rules $R$, i.e. with generators $V + X$ and relations $x \leq u$ for each $(x, u) \in R$.

\begin{definition}
  A \emph{preordered monoid} is a monoid $X$ equipped with a reflexive transitive relation $\leq$ such that $v \leq v'$ implies $u v w \leq u v' w$ for all $u, v, v', w \in X$.
\end{definition}

If we generalise this to arbitrary preordered monoids, i.e. to grammars with production rules that have arbitrary sources $R$ (and not just the ones with an $R$ with restrictions as defined above),  we obtain abstract rewriting systems, also known as semi-Thue systems \cite{Thue14,Power13}.
If we forbid rules with empty source or target we obtain the word-problem for semigroups, shown to be undecidable independently by Markov \cite{Markov47} and Post \cite{Post47}.
Both are equivalent in computational power to the last level of Chomsky's hierarchy: unrestricted grammars (which are thus undecidable).

\begin{definition}
  An \emph{unrestricted grammar} is a tuple $G = (V, X, R, s)$ where $V$ and $X$ are finite sets called the terminals and non-terminals with $R \sub (V + X)^+ \times (V + X)^*$ and $s \in X$.
  We define $L(G) = \set{u \in V^* \ \vert \ s \leq u}$ the language of $G$, where the order is given by the preordered monoid generated by $R$.
\end{definition}

Empirically, the complexity of human languages has been shown to lie somewhere in between that of context-free and unrestricted grammars \cite{Shieber87}, see the line of work on midly context-sensitive grammars \cite{Joshi85, Weir88, Vijay-ShankerWeir94}.

We cannot infer the underlying structure of human languages from complexity-theoretic arguments: indeed, the same class of languages can be generated by non-isomorphic grammar formalisms.
As our main example, Lambek's \emph{pregroup grammars} \cite{Lambek99, Lambek01, Lambek08} have the same expressive power as context-free grammars \cite{BuszkowskiMoroz08}, i.e. there are back-and-forth translations which preserve the generated languages.
However, pregroup grammars and context-free grammars are only weakly equivalent: these translations do not give an isomorphism.

\begin{definition}
  A \emph{pregroup}\footnote{
  The original definition of a pregroup was a partial, rather than pre, ordered monoid, i.e. one with an order relation with the anti-symmetry axiom. Anti-symmetry  causes unwanted equalities between types and   preorders were suggested as an alternative.}
   is a preordered monoid $P$ equipped with a pair of functions $(-)^l, (-)^r : P \to P$ such that $t^l t \leq \epsilon \leq t t^l$ and $t t^r \leq \epsilon \leq t^r t$ for all $t \in P$.
\end{definition}

\begin{definition}
  A \emph{pregroup grammar} is a tuple $G = (V, B, D, s)$ where $V, B$ and $D \sub V \times P_B$ are finite sets and $s$ is an element of $B$, for $B$ a set of basic grammatical types and $P_B$ the free pregroup generated by $B$.
\end{definition}

In the above, $V$ is the Vocabulary, $B$ a set of basic grammatical types such as $\{n,s\}$, for $n$ the type of a noun phrase and $s$ that of a sentence, and $D$ a relation known as the type dictionary or the lexicon. The language $L(G)$ of a pregroup grammar $G = (V, B, D, s)$ is the set of all strings $w_1w_2 \cdots  w_n \in V^*$ such that  $t_1 \dots t_n \leq s$ for each  $(w_i, t_i) \in D$.

\subsection{Monoidal Categories}

We assume familiarity with monoidal categories, see e.g. \cite{Awodey06} for an introduction.
In order to fix some notation, we define monoidal \emph{signatures} and \emph{presentations}.

\begin{definition}
  A monoidal \emph{signature} is a tuple $\Sigma = (\Sigma_0, \Sigma_1, \rm{dom}, \rm{cod})$ where $\Sigma_0$ and $\Sigma_1$ are sets of generating \emph{objects} and \emph{arrows} respectively, and $\rm{dom}, \rm{cod} : \Sigma_1 \to \Sigma_0^*$ are pairs of functions called \emph{domain} and \emph{codomain}.
\end{definition}

A homomorphism of monoidal signatures $h : \Sigma \to \Sigma'$ is given by a pair of functions $h_0 : \Sigma_0 \to \Sigma'_0$ and $h_1 : \Sigma_1 \to \Sigma_1'$ such that the following diagram commutes:

\begin{center}
\begin{tikzcd}[column sep=1in,row sep=1in]
\Sigma_0 \arrow{d}{h_0} & \Sigma_1 \arrow{l}{cod} \arrow{r}{dom}  \arrow{d}{h_1} & \Sigma_0 \arrow{d}{h_0} \\
\Sigma'_0 & \Sigma'_1 \arrow{l}{cod} \arrow{r}{dom}  & \Sigma'_0
\end{tikzcd}
\end{center}

We write $\bf{MonSig}$ for the category of monoidal signatures and their homomorphisms.
The free monoidal category $\bf{C}_\Sigma$ generated by a monoidal signature $\Sigma$ is the image of the left adjoint to the forgetful functor from $\bf{MonCat}$ to $\bf{MonSig}$.
Explicitly, the arrows $f : x \to y$ in $\bf{C}_\Sigma$ are given by (planar progressive) \emph{string diagrams} with arrows in $\Sigma_1$ as nodes and $x, y \in \Sigma_0^\star$ as input/output, see \cite{JoyalStreet88, JoyalStreet91, Selinger10}.
In order to impose equations on diagrams, we need some extra data: a presentation.

\begin{definition}
  A \emph{presentation} for a monoidal category is given by a monoidal signature $\Sigma$ and a set of relations $R \sub \coprod_{u , t \in \Sigma_0^*} \bf{C}_\Sigma(u, t) \times \bf{C}_\Sigma(u, t)$ between parallel arrows of the free monoidal category.
\end{definition}

The monoidal category generated by $(\Sigma, R)$ is the quotient $\bf{C}_\Sigma / R$, see \cite{Lane98}[II.8].

\subsection{Coalgebras}

Coalgebras are a useful framework for categorically modeling dynamical systems. In functional programming terminology, some functor $F$ defines a type signature for a mapping on $\mathbf{Set}$, and an $F$-coalgebra is a particular implementation of this signature.

\begin{definition}
For an endofunctor $F: \mathbf{Set} \rightarrow \mathbf{Set}$, an \emph{$F$-coalgebra} is a pair $(X, g)$, where $X$ is a set and $g : X \rightarrow FX$ is a function that defines structure over $X$.
\end{definition}

\begin{example}
For $X$ a set of states and  $F$  the identity functor,  an $F$-coalgebra $(X, g)$  maps states to states, computing a function over states such as permutations.
\end{example}

\begin{example}
The powerset functor $F(X) = {\cal P}(X)$ maps a state to a set of states, computing the accessibility relation of a Kripke structure or the transitions of an automata. 
\end{example}

\begin{example}
The functor $D(X) = \mathbb{B} \times X^{V}$ maps states in $X$ to the product of $\{\mathbf{false}, \mathbf{true}\}$ and the set of functions from some set $V$ back to $X$. This signature is the curried form of a deterministic automata transition function: each state in $X$ is tagged with a boolean value indicating whether or not it is an accept state, and each state defines a function from a vocabulary $V$ to the next state.
\end{example}

%% file: 3-monoidal-grammar.tex

\section{Monoidal Grammars}\label{grammars}

Our goal in this section will be to formally describe a characterization of formal grammars in terms of monoidal categories. We will begin with the following definition:

\begin{definition}
  A \emph{monoidal grammar} is a tuple $G = (\Sigma, R, s)$ where $\Sigma, R$ is the presentation of a monoidal category $\bf{C}_G$ with $V \cup \{s\} \in \Sigma_0$. We define
  $
  L(G) = \set{ u \in V^* \ \vert \ \exists \ f \in \bf{C}_G (u, s) }
  $.
\end{definition}

Given two monoidal grammars $G$ and $G'$ over a common vocabulary $V$, we will define the morphisms between them  to be the monoidal functors $F : C_G \to C_{G'}$ such that the following diagram commutes, where $inc$ and $inc'$ are the inclusion maps from $V^{*} \cup \{s\}$ into the objects of $\bf{C}_G$ and $\bf{C}_{G'}$ respectively:
\begin{center}
\begin{tikzcd}[column sep=1in,row sep=1in]
\bf{C}_G  \arrow{rd}{F} \\
V^{*} \cup \{s\} \arrow{u}{inc}  \arrow{r}{inc'} &  \bf{C}_{G'}
\end{tikzcd}
\end{center}

Monoidal grammars and their morphisms trivially form a subcategory of $(V + \set{s})^* / \bf{MonCat}$, which we will denote $\bf{Grammar}$. 

We can characterize the formal grammars that we described earlier as monoidal grammars. For example:
\begin{proposition}
  A pregroup grammar $G = (V, B, D, s)$ is a monoidal grammar with $\Sigma_0 = V + B \times \bb{Z}$, $\Sigma_1 = D + \coprod_{x \in B \times \bb{Z}} \set{ \tt{cup}_x, \tt{cap}_x }$\footnote{$\bb{Z}$ represents the number of adjoint applications, where negatives are left adjoints and positives are right adjoints.} and the relations $R$ given by the snake equations, see \cite{Selinger10}[4.1].
\end{proposition}

\begin{proposition}
  A context-free grammar $G = (V, X, R_G, s)$ is a monoidal grammar with $\Sigma_0 = V + X$, $\Sigma_1 = \{u \to x | \ (x, u) \in R_G\}$, and $R=\varnothing$.
\end{proposition}
We can now express the languages of these constructions in a unified manner:
\begin{definition}
For $w_1 w_2 ...  w_n \in V^{*}$ we call an arrow $w_1  w_2 ...  w_n \to r$ where $r \in Obj(\bf{C}_G)$ a \emph{parse state} of $w_1 w_2 ...  w_n$. Whenever in a parse state we have $r = s$, we call that parse state a  \emph{parsing} or a \emph{derivation}.
\end{definition}

\begin{definition}
The \emph{language} of a monoidal grammar is the set of all strings $w_1 \otimes w_2 ...  w_n$ that have a parsing.
\end{definition}


We can think of the string diagrams of parse states and parsings as akin to partially constructed and fully constructed syntax trees respectively.


\begin{example}
Let $G = (V, B, D, s)$ be a pregroup grammar with $V = \set{\rm{Alice}, \rm{loves}, \rm{Bob}}$, $B = \set{s, n}$ and $D = \set{(\rm{Alice}, n), \s (\rm{loves}, n^r s n^l), \s (\rm{Bob}, n)}$. Then the following string diagram represents the arrow in $\bf{C}_G$ from $Alice \otimes loves \otimes Bob$ to $s$:\\
\begin{center}
    \input{figures/pregroup-reduction}
\end{center}
\end{example}


\begin{example}
    Let $G = (V, X, R, s)$ be a context-free grammar with $V$ and $X$ defined as:
    \begin{align*}
        V = \set{\rm{Complex}, \rm{houses}, \rm{students}, \rm{disappoint}}\\
        X = \set{s, np, adj, vp, itv, tv}
    \end{align*}
    for $s$ the non-terminal representing a sentence, $np$ the non-terminal representing a  noun phrase, similarly, $adj$ an adjective, $itv$ an intransitive verb, $tv$ a transitive verb, and $vp$ a verb phrase. 
    Where the production rules are:
    \begin{align*}
        s \to np \ vp, \s vp \to tv \ np, \s vp \to itv,   \s np \to adj \ np\\
        np \to \text{complex}, \s adj \to \text{complex},\\
        \s np \to \text{houses}, \s itv \to \text{houses}, \\
        \s np \to \text{students}, \s itv \to \text{disappoint}
    \end{align*}
    Then the following string diagrams represent the arrows in $\bf{C}_G$ from $Complex \otimes houses \otimes students$ to $s$ and from $Complex \otimes houses \otimes disappoint$ to $s$:
    \input{figures/syntax-tree}
\end{example}


\subsection{$\mathbb{S}$-Monoidal Grammars}\label{weightedgrammars}
We can generalize monoidal grammars to associate weights (such as likelihoods or probabilities) with each parse state.

\begin{definition}
  For some semiring $\mathbb{S}$, an \emph{$\mathbb{S}$-monoidal signature} is a monoidal signature $\Sigma = (\Sigma_0, \Sigma_1, dom, cod)$ equipped with a mapping $r: \Sigma_1 \to \mathbb{S}$.
\end{definition}
A homomorphism of $\mathbb{S}$-monoidal signatures $h : (\Sigma, r) \to (\Sigma', r')$ is given by a monoidal signature homomorphism $h_0,h_1$ such that the following diagram commutes:

\begin{center}
\begin{tikzcd}[column sep=1in,row sep=1in]
\Sigma_1  \arrow{r}{r} \arrow{rd}{h_1} & \mathbb{S} \\
&  \Sigma_1'  \arrow{u}{r'}
\end{tikzcd}
\end{center}

\begin{definition}
  An \emph{$\mathbb{S}$-monoidal grammar} $(G,r)$ is a tuple $(\Sigma, R, s, r)$ such that $(\Sigma, r)$ is an $\mathbb{S}$-monoidal signature and $G = (\Sigma, R, s)$ is a monoidal grammar.
\end{definition}
We can extend $r$ into a function from $Ar(\bf{C}_G) \to \mathbb{S}$ by defining $r(g \circ f) = r(g)*r(f)$ and $r(f_1 \otimes f_2) = r(f_1)*r(f_2)$. If we do this, then $r$ defines a functor from $\bf{C}_G$ to the trivial monoidal category where arrows are the elements of $\mathbb{S}$ and both composition and tensor product are semiring multiplication.

We will refer to the category of $\mathbb{S}$-monoidal grammars over a common vocabulary as $\bf{Grammar}_\mathbb{S}$. The morphisms in this category are monoidal functors $F : \bf{C}_G  \to \bf{C}_{G'}$ such that the following diagram commutes:

\begin{center}
    \begin{tikzcd}[column sep=1in,row sep=1in]
    \bf{C}_G  \arrow{r}{r} \arrow{rd}{F}  & \mathbb{S}\\
    V^{*} \cup \{s\} \arrow{u}{inc}  \arrow{r}{inc'} &
    \bf{C}_{G'} \arrow{u}{r'}
    \end{tikzcd}
\end{center}

Let us note that we can characterize any monoidal grammar $G$ as the $\mathbb{B}$-monoidal grammar $(G, r)$ where $\mathbb{B}$ is the semiring of Booleans and  $r$ is the constant functor that maps all arrows in $\bf{C}_G$ to $true$ in $\mathbb{B}$.

%% file: figures/pregroup-reduction.tex
\begin{tikzpicture}[scale=0.55, baseline=(O.base)]

    \node (0) at (-3, 1) {};
    \node [scale=0.5, circle, fill=black] (1) at (-3, 0) {};
    \node [scale=0.5, circle, fill=black] (2) at (0, 0) {};
    \node (3) at (0, 1) {};
    \node (4) at (3, 1) {};
    \node [scale=0.5, circle, fill=black] (5) at (3, 0) {};
    \node (6) at (0, -2) {};
    \node (O) at (-3, 1.5) {Alice};
    \node (8) at (0, 1.5) {loves};
    \node (9) at (3, 1.5) {Bob};

    \node (13) at (-1.5, -1.5) {};
    \node (14) at (1.5, -1.5) {};

    \node (10) at (-3.5, -0.5) {$n$};
    \node (11) at (3.5, -0.5) {$n$};
    \node (12) at (0.5, -1.75) {$s$};
    \node (15) at (-0.75, -0.5) {$n^l$};
    \node (16) at (0.75, -0.5) {$n^r$};

		\draw (0.center) to (1.center);
		\draw (3.center) to (2.center);
		\draw (4.center) to (5.center);
		\draw (2.center) to (6.center);
		\draw [bend right=45] (1.center) to (13.center);
		\draw [bend right=45] (13.center) to (2.center);
		\draw [bend right=45] (2.center) to (14.center);
		\draw [bend right=45] (14.center) to (5.center);
\end{tikzpicture}

%% file: figures/syntax-tree.tex
$$\begin{tikzpicture}[scale=0.5]
		\node (0) at (-3, 0) {};
		\node (1) at (0, 0) {};
		\node (2) at (3, 0) {};
		\node [scale=0.5, circle, fill=black] (3) at (-3, 1) {};
		\node [scale=0.5, circle, fill=black] (4) at (0, 1) {};
		\node [scale=0.5, circle, fill=black] (5) at (3, 1) {};
		\node [scale=0.5, circle, fill=black] (6) at (1.5, 2.5) {};
		\node [scale=0.5, circle, fill=black] (7) at (0, 4) {};
		\node (8) at (-3, -1) {Complex};
		\node (9) at (0, -1) {houses};
		\node (10) at (3, -1) {students};
		\node (11) at (0, 5) {};
		\node (12) at (0, 6) {$s$};
		\node (13) at (2.75, 2) {$np$};
		\node (14) at (-1.25, 3.75) {$np$};
		\node (15) at (0.25, 2) {$tv$};
		\node (16) at (1.25, 3.75) {$vp$};

		\draw (11.center) to (7.center);
		\draw (7.center) to (6.center);
		\draw (6.center) to (4.center);
		\draw (6.center) to (5.center);
		\draw (7.center) to (3.center);
		\draw (3.center) to (0.center);
		\draw (4.center) to (1.center);
		\draw (5.center) to (2.center);
\end{tikzpicture}
\qquad \qquad
\begin{tikzpicture}[scale=0.5]
		\node (0) at (-3, 0) {};
		\node (1) at (0, 0) {};
		\node (2) at (3, 0) {};
		\node [scale=0.5, circle, fill=black] (3) at (-3, 1) {};
		\node [scale=0.5, circle, fill=black] (4) at (0, 1) {};
		\node [scale=0.5, circle, fill=black] (5) at (3, 1) {};
		\node [scale=0.5, circle, fill=black] (5) at (3, 2.5) {};
		\node [scale=0.5, circle, fill=black] (7) at (0, 4) {};
		\node (8) at (-3, -1) {Complex};
		\node (9) at (0, -1) {houses};
		\node (10) at (3, -1) {disappoint};
		\node (11) at (0, 5) {};
		\node (12) at (0, 6) {$s$};
		\node [scale=0.5, circle, fill=black] (13) at (-1.5, 2.5) {};
		\node (14) at (-1.25, 3.75) {$np$};
		\node (15) at (1.25, 3.75) {$vp$};
		\node (16) at (-2.5, 2.25) {$adj$};
		\node (17) at (-0.5, 2.25) {$np$};
		\node (15) at (2.25, 2.0) {$itv$};
		\draw (11.center) to (7.center);
		\draw (3.center) to (0.center);
		\draw (4.center) to (1.center);
		\draw (5.center) to (2.center);
		\draw (7.center) to (13.center);
		\draw (13.center) to (3.center);
		\draw (13.center) to (4.center);
		\draw (7.center) to (5.center);
\end{tikzpicture}
$$

%% file: 4-incremental-monoidal.tex
\section{Incremental Monoidal Grammars}\label{incrementalgrammars}

In an $\mathbb{S}$-monoidal grammar, we utilize the categorical composition to model the sequential nature of the parsing process and we use the tensor-product to model both the adjacency of words/types in a sentence and the parallel application of processing steps to different parts of a sentence. In this section we will explore how we can study the behavior of the tensor-product as a dynamic construct.

We can define an action of the free monoid $V^{*}$ on the category of endofunctors of $\bf{C}_G$ to model the process of adding a ``new'' word to our system. For any word $w \in V$, the endofunctor $W_w$ maps the object $o$ to $o \otimes w$ and the arrow $a$ to $a \otimes id_{w}$.
In order to model the interpretation of this new word in context, we look at the interaction of this endofunctor with the arrows in $\bf{C}_G$ by defining a mapping $W^{*}_w(a): Ar(\bf{C}_G) \rightarrow \mathbb{S}^{Ar(\bf{C}_G)}$ as follows.
\begin{align*}
    W^{*}_w(a)(a') = \begin{cases}
      r(a') &  a' \in X \\
      0 &  a' \not\in X
   \end{cases}\\
   X = \{a' \circ W_w(a)\  |\  a' \in Ar(C), dom(a')=(cod(a) \otimes w)\}
\end{align*}
Intuitively, $W^{*}_w(a)(a')$ is the generalized ``likelihood'' of the parse state $a' \circ W_w(a)$, given $W_w(a)$.
Naturally, parse states that cannot be expressed as $a' \circ W_w(a)$ have a likelihood of $0$ given $W_w(a)$.

If we define a ``listener'' to be an entity that assigns the parse state $a$ to the string of words $w_1  w_2 ...  w_n$, then $W^{*}_w$ defines how a new word may update a listener's parse state. We can represent this process with an $\mathbb{S}$-weighted automaton. The states in the machine are parse states $a: w_1  w_2 ...  w_n \rightarrow o$, and the transition function is $W^{*}_w$. $\mathbb{S}$-weighted automaton are also equipped with a function $r_0$ that maps states to $\mathbb{S}$. For example, in the $\mathbb{B}$ case $r_0$ picks out the ``accept states'' of the automaton. For some $\mathbb{S}$-monoidal grammar $(\bf{C}_G, r)$ we will define this function $r_0$ as follows:
\begin{align*}
    r_0(a) = \begin{cases}
      r(a) &  cod(a) = s \\
      0 &   cod(a) \neq s
   \end{cases}
\end{align*}





In order to formalize this relationship,  we will use the theory of $F$-coalgebras on $\mathbf{Set}$  \cite{Rutten:1996:UCT:869662}. An $\mathbb{S}$-weighted automaton is a coalgebra of the functor $\mathcal{W}(X) = \ssxv$ and by \cite{SilvaEtAl13} a given transition function will uniquely define coalgebras of $\mathcal{W}$ up to bisimulation, so we simply need to define a mapping $\mathcal{I}$ between $\bf{Grammar}_\mathbb{S}$ and coalgebras of $\mathcal{W}$.

\begin{definition}
The \emph{Incremental Functor} $\mathcal{I}$ is a map on the objects and arrows of $\bf{Grammar}_\mathbb{S}$ that acts as follows:
\begin{itemize}
    \item For some object $(\bf{C}_G, r)$ in $\bf{Grammar}_\mathbb{S}$, we define $\mathcal{I}(\bf{C}_G, r)$ to be the coalgebra $(Ar(\bf{C}_G), r_0 \times W^{*}_w)$.
    \item For some morphism $F$ between $(\bf{C}_G, r_0)$ and $(\bf{C}_G', r'_0)$ in $\bf{Grammar}_\mathbb{S}$, we define $\mathcal{I}(F)$ to be the map $h_F: Ar(\bf{C}_G) \rightarrow Ar(\bf{C}_G')$ that uses $F$'s mapping on arrows.
\end{itemize}
\end{definition}

\begin{lemma}\label{incrementalproof}
The map $\mathcal{I}$ is a functor from $\bf{Grammar}_\mathbb{S}$ to the category of coalgebras of $\mathcal{W}$.
\end{lemma}
\begin{proof}
Since $\mathcal{I}(F)$ uses the functor $F's$ mapping on arrows, $\mathcal{I}$ trivially preserves identity and composition. All that remains is to prove that $\mathcal{I}(F) = h_F$ is indeed a coalgebra homomorphism. To do this we need to show that
the following diagram commutes:

\newcommand{\rw}{r_0 \times W^{*}_w}
\newcommand{\rwp}{r'_0 \times W'^{*}_w}
\newcommand{\ssh}{\mathbb{S} \times (\mathbb{S}^{h_F})^{V}}

\begin{center}
\begin{tikzcd}[column sep=1in,row sep=1in]
Ar(\bf{C}_G) \arrow{d}{\rw} \arrow{r}{h_F} & Ar(\bf{C}_G') \arrow{d}{\rwp} \\
\mathbb{S} \times (\mathbb{S}^{Ar(\bf{C}_G)})^{V}  \arrow{r}{ \ssh} &  \mathbb{S} \times (\mathbb{S}^{Ar(\bf{C}_G')})^{V}
\end{tikzcd}
\end{center}
When we apply the function $\ssh$ to $\rw$ we get the following function:
\begin{align*}
      \ssh(\rw)(a) = r'_0(h_F(a)) \times \lambda a' . \begin{cases}
      r'(h_F(a')) &  h_F(a') \in X' \\
      0 &  h_F(a') \not\in X'
   \end{cases} \\
   X' = \{h_F(a' \circ (a \otimes id_{w}))\  |\  a' \in Ar(\bf{C}_G'), dom(a')=(cod(a) \otimes w)\}
\end{align*}
Here $r'_0(h_F(a)) = r_0(a)$ because $F$ and $r$ commute in $\bf{Grammar}_\mathbb{S}$. When we use the fact that $h_F$ preserves composition and tensor product
we can rewrite $X'$ as:
\begin{align*}
     X' = \{h_F(a') \circ (h_F(a) \otimes id_{w})\  |\  h_F(a') \in Ar(\bf{C}_G'), dom(h_F(a'))=(cod(h_F(a)) \otimes w)\}
\end{align*}
This makes it clear that:
\begin{align*}
W'^{*}_w(h_F(a))(a') = \begin{cases}
  r'(h_F(a')) &  h_F(a') \in X' \\
  0 &  h_F(a') \not\in X'
\end{cases}
\end{align*}
Therefore, we can conclude that:
\begin{align*}
\ssh(\rw) = (\rwp) \circ h_F
\end{align*}
So $h_F$ is indeed a coalgebra homomorphism.

\end{proof}



The incremental functor reveals an interesting duality. Given a monoidal category that describes how laterally composed (tensor product) objects are processed vertically (by composition of arrows), we get for free an automaton that defines a laterally evolving process upon vertically composed arrows. This automaton models how the process that the monoidal category describes will behave in the face of new information, and it reveals the monoidal category's underlying incremental structure. This suggests the following corollary:
\begin{corollary}
If the $\mathbb{S}$-monoidal grammars $(\bf{C}_G,r)$ and $(\bf{C}_G',r')$ have functors between them, then the coalgebras $\mathcal{I}(\bf{C}_G,r)$ and $\mathcal{I}(\bf{C}_G',r')$ are bisimulatable.
\end{corollary}
\begin{proof}
If $(\bf{C}_G,r)$ and $(\bf{C}_G',r')$ have functors between them, then $\mathcal{I}(\bf{C}_G,r)$ and $\mathcal{I}(\bf{C}_G',r')$ have coalgebra homomorphisms between them. By \cite{levycoalgebra}, this implies that $\mathcal{I}(\bf{C}_G,r)$ and $\mathcal{I}(\bf{C}_G',r')$ are bisimulatable.
\end{proof}

%% file: 5-r-monoidal-grammars.tex
\section{$\mathbb{R}_{\geq 0}$-Monoidal Grammars from Language Models}\label{rgrammars}


For some $\mathbb{R}_{\geq 0}$-monoidal grammar $(G, r)$, we can think of $r$ as assigning ``confidence values'' or ``likelihoods'' to the arrows in $G$. This allows us to model the relative probability of different interpretations of the same string of words and incorporate probabilistic and distributional language models into the parsing process.
Let us explore what this means.

\subsection{Maximum Parse State Likelihood}
\begin{definition}
A \emph{language model} over the monoidal grammar $G$ is the following family of probability distributions:
\begin{itemize}
    \item For any prefix $w_1 \otimes ... w_n$, a probability distribution with finite support over the possible completions $w_{n+1} \otimes ... w_m$ such that there is a parsing $w_1 \otimes ... w_n \otimes w_{n+1}  \otimes ... w_m \to s$. We will denote this as $Pr( w_{n+1} \otimes ... w_m | w_1 \otimes ... w_n)$.
    \item For any sentence $w_1 \otimes ... w_m$ with at least one parsing $a: w_1 \otimes ... w_m \to s$, a probability distribution  with finite support over the set of all such parsings. We will denote this as $Pr(a | w_1 \otimes ... w_m)$.
\end{itemize}

\end{definition}
We would like to develop a mapping from language models over a monoidal grammars $G$ to $\mathbb{R}_{\geq 0}$-monoidal grammars $(G,r)$ such that $r(p)$
is indicative of the probability that a particular parse state is ``optimal'' according to that language model. We can formalize this notion with the following concepts:
\begin{definition}
For some prefix $w_1 \otimes ... w_n$ equipped with a parse state $p: w_1 \otimes ... w_n \rightarrow o$ and sentence $w_1 \otimes ... w_n \otimes w_{n+1}  \otimes ... w_m$ equipped with a parsing $a: w_1 \otimes ... w_n \otimes w_{n+1}  \otimes ... w_m \to s$, we say that $a$ is \emph{$p$-compliant} if there exist some arrows $a_h: w_{n+1}  \otimes ... w_m \rightarrow o', a_v: o \otimes o' \rightarrow s$ such that $a = a_v \circ (p \otimes a_h)$.
\end{definition}
\begin{definition}
Given some prefix $w_1 \otimes ... w_n$ and sentence $w_1 \otimes ... w_n \otimes w_{n+1} \otimes ... w_m$ such that there exists a parsing $a: w_1 \otimes ... w_n \otimes w_{n+1} \otimes ... w_m \rightarrow s$, we will define a \emph{maximal parse state} for $w_1 \otimes ... w_n$ on $a$ to be a parse state $p: w_1 \otimes ... w_n \rightarrow o$ such that $a$ is $p$-compliant and there exists no $c$ such that $c \circ p$ is not isomorphic to $p$ and $a$ is ($c \circ p$)-compliant.
\end{definition}
Intuitively, the parsing $a$ is $p$-compliant if the parse state $p$ factors $a$, and $p$ is a maximal parse state for $a$ if $p$ cannot be developed any farther while still maintaining the possibility to extend it into $a$. We can use these constructs to define the following quantity:
\begin{definition}
Given a parse state $p: w_1 \otimes ... w_n \rightarrow o$, its \emph{maximal parse state likelihood} $pr_M(p)$ is:
\begin{align*}
    pr_M(p) = \sum_{a\in A_{p}} Pr(a | w_1 \otimes ... w_n \otimes w_{n+1} \otimes ... w_m) Pr(w_{n+1} \otimes ... w_m | w_1 \otimes ... w_n)
\end{align*}
Where $A_{p}$ is the set of all parsings $a$ such that $p$ is a maximal parse state for $w_1 \otimes ... w_n$ on $a$.
\end{definition}
Intuitively, given a prefix $w_1 \otimes ... w_n$, a parse state $p: w_1 \otimes ... w_n \rightarrow o$, and a language model over the possible fully parsed sentences, $pr_M(p)$ is the probability that $p$ will be the maximal parse state for the completed sentence.

Now say we want to construct an $\mathbb{R}_{\geq 0}$-monoidal grammar $(G,r)$ such that for any parse state, $r(p)$ is as close as possible to $pr_M(p)$. By the definition of $\mathbb{S}$-monoidal grammars, $r$ is defined by the values of $r(g)$ for all generating arrows $g \in \Sigma_1$. For example, for any parse state $p$ in $G$, we have that:
\begin{align*}
    r(p) = \prod_{g \in B_{p}} r(g)
\end{align*}
Here $B_{p}$ is any bag of generating arrows in $\Sigma_1$ that form $p$ when combined by composition and tensor product.

Therefore, our objective is to find $r$ to minimize the following expression, where $P_G$ is the set of all parse states in $G$:
\begin{align*}
    min_{r} \norm{\sum_{p \in P_G} \left(
    pr_M(p) - \prod_{g \in B_{p}} r(g)
    \right)}
\end{align*}
Because $log$ is a monotonic function, we can rewrite this as the following, where $B_{p}(g)$ is the number of times that $g$ appears in $B_{p}$:
\begin{align*}
    min_{r} \norm{\sum_{p \in P_G} \left(
    log(pr_M(p)) - \sum_{g \in \Sigma_1} log(r(g)) * B_{p}(g)
    \right)}
\end{align*}
This is a standard linear approximation of the coefficients $log(r(g))$. If $P_G$ is finite, we can use the normal equations or gradient descent to learn the values of $r(g)$. If not, we can use stochastic gradient descent \cite{Bishop:2006:PRM:1162264}.




%% file: 999-future-works.tex

\section{Future Work}\label{futurework}
We will close with some remarks on future directions.
\begin{itemize}
    \item In section \ref{rgrammars} we describe a simple algorithm for constructing an $\mathbb{R}_{\geq 0}$-monoidal grammar from a language model over a monoidal grammar. However, since we do not describe a method for constructing such language model from a dataset and choice of monoidal grammar, it is not immediately clear how we can implement this algorithm in practice.
    Such a method would enable us to efficiently implement our algorithm on top of any linguistic parsing library for CFGs, pregroups, or other monoidal grammars.

    \item In section \ref{incrementalgrammars} we introduce a functor $\mathcal{I}$ from $Grammar_{\mathbb{S}}$ to the coalgebras of the functor $\mathcal{W} = \ssxv$. A natural question is whether there are functors in the reverse direction as well. In particular, we are curious as to whether $\mathcal{I}$ has left/right adjoints. If such functors exist, this would be further evidence of a fundamental connection between dynamical systems and monoidal categories. Such a discovery could help us unite these fields of study, as well as enable us to reduce questions about automata into questions about monoidal categories.

    \item Our framework is fully syntax-based and is currently agnostic to semantic meaning.
    In \cite{DBLP:journals/corr/abs-1003-4394}, the authors define a functorial framework for affiliating syntactic structure with vector space semantics, and in \cite{SadrzadehEtAl18a} the authors project vector space semantics onto dynamic syntax trees in order to model incrementality. This raises the question of how we can best incorporate semantic information into our framework. For example, if we extend the functor from \cite{DBLP:journals/corr/abs-1003-4394} to operate over monoidal grammars in general, it would be interesting to see how the interaction between this functor and $\mathcal{I}$ compares to the construction in \cite{SadrzadehEtAl18a}.

\end{itemize}